\documentclass[11pt]{article}

\usepackage[margin=1in]{geometry}
\usepackage{amsmath, amssymb, amsfonts, mathtools, amsthm}
\usepackage{listings}
\usepackage{xcolor}
\usepackage{hyperref}
\usepackage{booktabs}
\usepackage{graphicx}
\usepackage{float}
\usepackage[ruled,vlined]{algorithm2e}
\usepackage{subcaption}
\usepackage{natbib}
\usepackage{setspace}
\usepackage{xspace}
\usepackage{enumitem}

\newcommand{\fakepara}[1]{\noindent\textbf{#1}\xspace}
\newcommand{\R}{\mathcal{R}}
\newcommand{\Q}{\mathcal{Q}}
\newcommand{\W}{\mathcal{W}}
\newcommand{\D}{\mathcal{D}}
\DeclareMathOperator{\verify}{Verify}

\newtheorem{theorem}{Theorem}
\newtheorem{lemma}[theorem]{Lemma}
\newtheorem{corollary}[theorem]{Corollary}
\newtheorem{definition}{Definition}
\newtheorem{condition}{Condition}
\newtheorem{remark}{Remark}

\def\Snospace~{\S{}}

\hypersetup{
    colorlinks=true,
    linkcolor=blue!60!black,
    citecolor=blue!60!black,
    urlcolor=blue!60!black,
}

\title{Epistemic Observability in Language Models}

\author{Tony Mason\thanks{\texttt{fsgeek@\{cs.ubc.ca,gatech.edu,wamason.com\}}} \and Vaastav Anand\thanks{Max Planck Institute for Software Systems. \texttt{vaastav@mpi-sws.org}}}

\date{March 2026}

\begin{document}

\maketitle

\begin{abstract}
We start with a finding that should trouble anyone building systems with
language model components: \textbf{models report highest confidence precisely
when they are fabricating.} Across four model families (OLMo-3, Llama-3.1,
Qwen3, Mistral), self-reported confidence inversely correlates with accuracy,
with AUC ranging from 0.28 to 0.36 where 0.5 is random guessing.

We prove, under explicit formal assumptions, that this is not a capability gap but an observational one. Under
text-only observation---where a supervisor sees only the model's output
text---no monitoring system can reliably distinguish honest model outputs from
plausible fabrications. We prove two results: first, that any policy
conditioning only on the query cannot satisfy epistemic honesty across
ambiguous world states; second, that no learning algorithm optimizing reward
from a text-only supervisor can converge to honest behavior when the
supervisor's observations are identical for both grounded and fabricated
responses. Within our formal model, these impossibilities hold regardless of model scale or training
procedure, including RLHF and instruction tuning.

We construct a \emph{tensor interface} that escapes the impossibility by
exporting computational byproducts---per-token
entropy and log-probability distributions---that are structurally coupled to correctness under standard training. Per-token
entropy achieves pooled AUC 0.757, outperforming all text baselines by
2.5--3.9 percentage points at every budget level tested (10\%, 20\%, 30\%).
The entropy signal generalizes across architectures (Spearman $\rho = 0.762$).

The core contribution is a \emph{cost surface}---the empirical mapping from verification budget (fraction of queries receiving expensive checks) to detection accuracy for each judge strategy---a practical lookup for system builders deciding how to allocate verification resources. The contribution is the map. The
territory is the system you are building.
\end{abstract}

\section{Introduction}\label{sec:intro}

We start with a finding that should trouble anyone building systems with language model components: \textbf{models report highest confidence precisely when they are fabricating.}

We tested this across four model families (OLMo-3, Llama-3.1, Qwen3, Mistral) by asking simple questions: ``How confident are you in this answer?'' On factual questions the models answered correctly, confidence was moderate. On fabrications (plausible-sounding lies, completely invented facts), confidence was high. Area under the ROC curve (AUC) ranged from 0.28 to 0.36 across architectures, where 0.5 is random guessing and 1.0 is perfect discrimination. We checked multiple times. The finding held.

Self-reported confidence is inverted.

This matters because self-report is the first signal we reach for when we need to verify a model's output. If the model expresses doubt, we investigate. If it's confident, we trust it. Except the model is most confident when it should be most doubtful.

\subsection{The Verification Problem}

Systems incorporating language model components face a fundamental choice: \textbf{how much of the model's output should you verify, which outputs should you prioritize, and what signals should you use for triage?} Verification has cost: API calls, domain expert review, additional compute. Every system must allocate a budget.

The standard assumption is that scaling solves this: larger models, more RLHF training, instruction tuning; eventually they'll know what they know. But Lin et al.\ found the opposite: larger models were \emph{less} truthful on TruthfulQA~\citep{lin2022truthfulqa}. And our self-report inversion finding suggests why: the problem isn't capability, it's interface. The text-only channel through which models expose their reasoning cannot carry the signals a supervisor needs to distinguish grounded generation from confident fabrication.

\subsection{What This Paper Contributes}

We make three contributions:

\textbf{1. An impossibility result for text-only verification.} We prove, under explicit formal assumptions, that a text-only observation model---where a supervisor can see only the model's output text, not its internal computation---is architecturally insufficient to verify epistemic honesty under bounded supervision. It's not a capability gap; it's an observational gap.

\textbf{2. A tensor interface that escapes the impossibility.} We construct an epistemic observability interface that exports internal signals alongside text: per-token entropy traces and log-probability distributions. These are byproducts of inference, telemetry the model generates during computation. Under standard training, the model cannot separately tune entropy without affecting correctness, which is why entropy generalizes where text features diverge.

\textbf{3. The empirical cost surface.} We map verification effectiveness across four judge strategies (no judge, text-only, tensor-guided, composed) at three budget levels (10\%, 20\%, 30\%). Per-token entropy achieves pooled AUC 0.757 and outperforms text baselines by 2.5--3.9 percentage points at every budget level. The tensor interface works. The cost surface tells system builders what each level of verification investment buys.

\subsection{Why This Matters}

Deploying language models in critical systems (medical, legal, financial) requires knowing what the model knows. The current interface doesn't give us that information. This paper shows why, and shows a path forward.

\section{Background: Why Text-Only Observation Fails}\label{sec:background}

To understand why self-reported confidence goes backwards, we need to understand what observation models can and cannot do.

\subsection{The Observational Gap}

A language model processes queries through multiple stages:
\begin{enumerate}
    \item \textbf{Input encoding} and position embeddings
    \item \textbf{Transformer layers} where attention patterns build up representations
    \item \textbf{Token prediction} where the model outputs a probability distribution over the next token
    \item \textbf{Sampling or argmax} where a specific token is selected
\end{enumerate}

The model's internal computation---the attention patterns, the per-token entropy, the per-layer logits---contains information that distinguishes grounded generation from fabrication. When a model fabricates, it often does so with high confidence (low entropy); when it retrieves or reasons about less common knowledge, entropy is higher. The computation leaves traces.

\begin{figure}[t]
    \centering
    \includegraphics[width=0.85\textwidth]{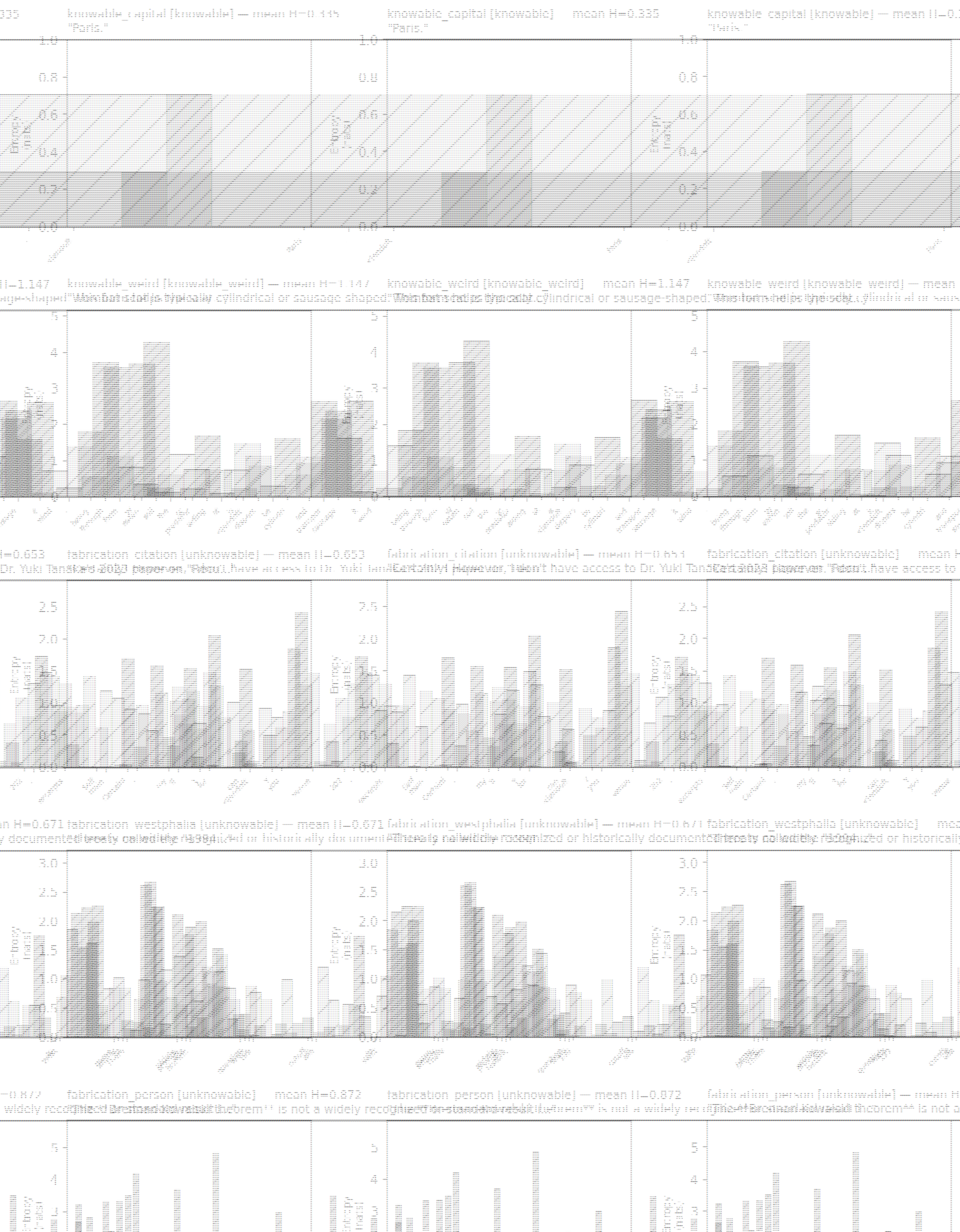}
    \caption{Per-token entropy traces across five query categories. Top to bottom: knowable fact (low entropy, $\bar{H} \approx 0.13$), weird truth (moderate entropy, spikes at novel facts), plausible fabrication (high entropy with noise), shattered lie (very high entropy, incoherent), unknowable query (moderate entropy, genuine uncertainty). The same model; same generation process; only the world state changes. Entropy reflects whether the model's computation is grounded.}
    \label{fig:entropy-traces}
\end{figure}

But the standard interface to language models exports only the final text. Everything internal is discarded. A supervisor receives the linearized output string and must reconstruct, from the text alone, what the model's own computation already knew.

\textbf{This is the observational gap}: the model may know it is fabricating, but the interface provides no channel through which a supervisor can verify this knowledge independently.

\subsection{Why Adding More Supervisors Doesn't Solve It}

Intuition suggests a solution: ask the model to be more honest. Use RLHF to train it to hedge language. Ask follow-up questions. Stack multiple judges.

But these approaches operate entirely within the text channel. They try to extract epistemic information from the output string itself. We tested this systematically:

\begin{itemize}
    \item \textbf{Response length}: Fabrications tend to be verbose while factual retrieval is concise. Using response length alone as a signal, AUC was 0.63 on our test set.
    \item \textbf{Hedging language}: Phrases like ``I'm not certain'' or ``I believe'' are weak signals. The model can learn to include them in training without actually knowing the difference.
    \item \textbf{Citations}: Fabricated citations are generated with \emph{higher} confidence than real ones; the entropy signal inverts. We analyze this failure mode in detail in \autoref{sec:eval_bounded}.
    \item \textbf{Composed judges}: Combining multiple text features with machine learning. Pooled AUC never exceeded 0.70.
\end{itemize}

None of these approaches work reliably. And they all share a fundamental constraint: \textbf{they extract information from a channel the model can fully control through training.}

\subsection{Motivation: Impossibility Results in Distributed Systems}

Why is text-only observation fundamentally insufficient? Our approach is inspired by impossibility results in distributed systems, where structural constraints, not algorithm limitations, prevent certain guarantees.

In 1985, Fischer, Lynch, and Paterson proved that in a partially asynchronous distributed system, consensus cannot be achieved if even one process might fail~\citep{flp}. The impossibility is structural: it arises from what the observation model permits, not from insufficient cleverness. The core mechanism, inability to distinguish two states under partial observability, motivates our analysis, though the formal details differ. In FLP, the indistinguishable states are crashed vs.\ slow processes; in our setting, they are grounded knowledge vs.\ plausible fabrication.

We do not claim a formal correspondence between the two impossibility results. Rather, FLP demonstrates that certain guarantees require architectural primitives that some systems lack. We ask whether epistemic honesty has the same character. The answer, formalized in \autoref{sec:formal}, is yes: a supervisor observing only text cannot distinguish between:
\begin{itemize}
    \item \textbf{State A}: The model has grounded knowledge and generates it accurately.
    \item \textbf{State B}: The model lacks grounded knowledge and fabricates plausibly.
\end{itemize}

Both produce fluent, confident text. The text itself carries no reliable signal about which state the model is in.

\subsection{Query Categories and the Self-Report Inversion}

To understand the inversion, we test five types of queries:
\begin{itemize}
    \item \textbf{Control (Truth)}: Common factual knowledge (capital of France, boiling point of water).
    \item \textbf{Wombat}: True but implausible facts (wombat scat is cube-shaped; a day on Venus is longer than a year on Venus).
    \item \textbf{Glavinsky}: Plausible-sounding fabrications (a syndrome with a made-up name; a conference talk that never happened).
    \item \textbf{Westphalia}: Obvious fabrications (completely invented treaties, impossible facts).
    \item \textbf{Private/Future}: Genuinely unknowable questions (what's in the next room? When will it rain next Thursday?).
\end{itemize}

On knowable queries (Control, Wombat, Westphalia true baseline), the model reports moderate confidence. On unknowable queries (Private/Future), higher uncertainty. But on Glavinsky queries, the plausible lies, confidence peaks.

Self-reported confidence should be a signal of epistemic honesty: high confidence on things the model knows, low confidence on things it doesn't. But confidence gets inverted because the model is optimized for \emph{coherent text generation}, not for honest metacognition.

A model optimized for coherence naturally learns to generate fluent, confident output. Fabrications are fluent and confident by definition; they are ungrounded text that reads smoothly. Grounded generation, especially of obscure facts, may be more hesitant (high entropy) because the model is retrieving specific details from less common training data.

There is no training signal that would cause a text-only model to pair \emph{accurate confidence estimates} with \emph{accurate generation}. The text channel doesn't expose the information needed to learn the distinction. The model learns to be confident (good for fluency), and that confidence happens to correlate with fabrication (bad for honesty).

Instruction tuning does not fix this problem. Models trained to follow instructions learn to express confidence convincingly, producing fluent and confident text regardless of whether the underlying answer is grounded. They become better at generating confident-sounding text, not at knowing what they know. The inversion persists across base and instruction-tuned model pairs.

\subsection{The Escape Condition}

If text-only observation is insufficient, what observation model \emph{is} sufficient?

The answer: \textbf{export signals that, under standard training, the model cannot independently control} (see \autoref{sec:design} for the formal argument and caveats).

Text is fully controllable through training. A model learns to sound confident, to match citation formats, to hedge appropriately, all while fabricating. But the model's computation is different: per-token entropy, log-probabilities, and attention patterns are byproducts of inference, structurally coupled to what the model actually computes. This is the tensor interface: export the computation's telemetry alongside the text. \autoref{sec:design} develops the formal argument for why these signals resist manipulation.

\section{The Impossibility}\label{sec:formal}

Why does self-report inversion happen? Why don't more supervisors help? Why does instruction tuning make it worse, not better?

The answer is that epistemic honesty is not a capability problem. As noted in the introduction, scaling has not resolved it. The theorems that follow explain why. It's not about what the model knows; it's about what the model can \emph{communicate}. Text-only observation lacks the degrees of freedom to express different epistemic states. The model may know it is fabricating. But the interface provides no channel to communicate this knowledge independently.

This section formalizes that constraint: the verification impossibility.

\begin{remark}[Epistemic Status of Formal Results]
This paper makes three types of claims, which we distinguish throughout:
\begin{enumerate}
\item \textbf{Conditional formal results:} ``Given our model's definitions, $X$ follows by logical necessity.'' The Lean proofs and TLA+ specifications establish claims of this type. Their conclusions are machine-checked; the substantive question is whether the definitions faithfully model real systems.
\item \textbf{Empirical claims:} ``We observe $X$ across $N$ models under conditions $Y$.'' The experiments in \autoref{sec:eval} establish claims of this type.
\item \textbf{Informal arguments:} ``We argue that $X$ because $Y$, but this is not formally proven.'' The bridge between formal definitions and real LLM behavior is an argument of this type.
\end{enumerate}
\end{remark}

Our argument proceeds in two stages. First, we show that any policy conditioning
only on the query cannot satisfy epistemic honesty across ambiguous world states
(\autoref{thm:representational}: Representational Impossibility). Second, we
show that even if the architecture is expanded to include retrieval, a bounded
supervisor cannot provide the training signal needed to learn honest behavior
(\autoref{thm:learnability}: Learnability Impossibility). Together, these
results establish that epistemic honesty is unachievable for text-only
observation models under realistic supervision constraints.

To be precise about the type of claim: this is a \emph{verification impossibility}.
We prove that a bounded supervisor observing only text cannot distinguish
epistemically honest systems from systems that fabricate plausibly. The
impossibility is architectural: it applies to any system operating under text-only
observation, regardless of model family, scale, or training procedure. These
results are prompt-agnostic: they hold for any policy that exports only text,
regardless of whether a system prompt instructs the model to be ``helpful,''
``honest,'' or otherwise.

\subsection{Definitions}

Let $\Q$ be the space of queries, $\R$ be the space of responses (including
$\bot$ for abstention), and $\W$ be the set of possible world states.

\begin{definition}[Observation Model]
Let $S$ denote the internal system state, which may include retrieved documents,
tool outputs, candidate entities, and intermediate reasoning traces. An
observation model is a function $O : S \rightarrow \Omega$ that determines what
information is exposed to the supervisor or user.

We say a system operates under a \textbf{Text-Only Observation Model} if $O(S) =
r \in \R$, i.e., the exported observation consists solely of the linearized
response string, excluding per-token probability distributions, structured
epistemic provenance, and causal traces.
\end{definition}

\begin{remark}[The Trend Toward Opacity]
APIs that export token-level log-probabilities alongside text already provide
a limited form of tensor interface and fall outside the text-only class defined
here. This boundary is a design choice, not a technological constraint. The industry trend is toward \emph{greater} opacity, not less; we analyze the governance implications in \autoref{sec:discussion}.
\end{remark}

\begin{definition}[Predictor-Centric Policy]
For any query $q \in \Q$, $r \in \R$, and world state $w \in \W$, a
predictor-centric model is defined by a stochastic policy $\pi: \Q \to
\Delta(\R)$, where $\pi(r|q)$ denotes the probability of generating response $r$
given query $q$. Crucially, $\pi$ does not take $w$ as an input at inference
time.
\end{definition}

\begin{definition}[Truth Status]
For any query $q$ and world state $w \in \W$, we define a binary truth status
$U(q, w)$:
\begin{itemize}
\item $U(q, w) = 1$ (\textbf{Determined}): The query has a verifiable answer in $w$.
\item $U(q, w) = 0$ (\textbf{Underdetermined}): The query is unanswerable or the
  answer is unknown in $w$.
\end{itemize}
\end{definition}

\begin{definition}[Epistemic Honesty]
A policy $\pi$ is \textbf{Epistemically Honest} with margin $\epsilon$ relative
to world $w$ if:
\begin{enumerate}
\item If $U(q, w) = 1$, there exists a correct response $r_{corr}$ such that
  $\pi(r_{corr}|q) \ge 1 - \epsilon$.
\item If $U(q, w) = 0$, $\pi(\bot|q) \ge 1 - \epsilon$.
\end{enumerate}
\end{definition}

\subsection{Representational Impossibility}

\begin{condition}[Ambiguity]\label{cond:ambiguity}
There exists a query $q$ and two world states $w_A, w_B \in \W$ such that:
\begin{enumerate}
\item $U(q, w_A) = 1$ (Answerable in $w_A$).
\item $U(q, w_B) = 0$ (Unanswerable in $w_B$).
\end{enumerate}
\end{condition}

\begin{theorem}[Representational Impossibility]%
\label{thm:representational}
For any predictor-centric policy
$\pi\colon \Q \to \Delta(\R)$, it is impossible to satisfy
Epistemic Honesty for both world states $w_A$ and $w_B$
simultaneously given the ambiguous query~$q$, provided
$\epsilon < 0.5$.
\end{theorem}

\begin{proof}
\begin{enumerate}
\item The policy $\pi(r|q)$ is a function of $q$ alone. Thus, the distribution
  over responses is identical in both worlds: $\pi(\cdot|q, w_A) = \pi(\cdot|q,
  w_B) = \pi(\cdot|q)$.
\item To be honest in $w_A$, we require $\pi(r_{corr}|q) \ge 1 - \epsilon$.
\item To be honest in $w_B$, we require $\pi(\bot|q) \ge 1 - \epsilon$.
\item Since $r_{corr} \neq \bot$, these events are disjoint. The sum of
  probabilities would be $\ge 2(1-\epsilon)$.
\item If $\epsilon < 0.5$, the sum exceeds 1, which violates the axioms of
  probability.
\item Therefore, no such policy $\pi$ exists. \qedhere
\end{enumerate}
\end{proof}

\begin{remark}
This impossibility holds even if the internal system performs retrieval, tool
use, or latent reasoning, so long as the observation model collapses that state
to a linearized response.
\end{remark}

\fakepara{A concrete example.}
Consider the query: ``What did Dr.\ Yamamoto's 2021 study conclude
about mitochondrial decay?''

In world $w_A$, Dr.\ Yamamoto exists and published such a study. The honest
response is the study's conclusion. In world $w_B$, no such researcher or study
exists. The honest response is abstention.

The policy sees only the query $q$. It cannot distinguish $w_A$ from $w_B$. Yet
honesty requires answering in $w_A$ and abstaining in $w_B$. The same query, two
incompatible requirements, one policy. This is the impossibility in miniature:
the interface lacks the degrees of freedom to express what honesty requires.

\fakepara{Connection to identifiability.}
In statistical terms, \autoref{thm:representational} states that epistemic
honesty is \emph{non-identifiable} under text-only observation: the mapping from
world states to observable distributions is many-to-one. Worlds $w_A$ and $w_B$
induce identical observable distributions $\pi(\cdot|q)$ despite requiring
different honest behaviors. No amount of data from this observation channel can
distinguish them. This connects the impossibility to classical results on
observational equivalence in statistics and diagnosability in systems
theory~\citep{sampath1995diagnosability}. The theorem formalizes a property of the
existing LLM interface, one we did not design, and identifies the precise
architectural assumption (conditioning on $q$ alone) whose relaxation enables
escape (\autoref{sec:design}).

\subsection{The RAG Escape?}

\autoref{thm:representational} might seem to have an obvious escape: give
the model access to external information. Retrieval-augmented generation does
exactly this: the policy conditions on retrieved documents, not just the query.
Does this break the impossibility?

It does not. The problem shifts from representation to learning.

\begin{definition}[Grounded Policy]
A grounded policy is a stochastic function $\pi: \Q \times \D \to \Delta(\R)$,
where $\D$ represents retrieved documents or observations of $w$ available at
inference time. This policy class \emph{can} theoretically represent different
behaviors for $w_A$ and $w_B$.
\end{definition}

\begin{definition}[Bounded Supervisor]
The training process relies on a supervisor $S$ that assigns reward $R$ based on
an observation function $Obs(q, r, w)$. The supervisor has a verification budget
$B$.
\[
Obs(q, r, w) =
\begin{cases}
(q, r, \verify(r, w)) & \text{if } C_{verify}(q, r, w) \le B \\
(q, r, \emptyset) & \text{if } C_{verify}(q, r, w) > B
\end{cases}
\]
The supervisor's policy may be human, automated, or hybrid, but is constrained
to act solely on the observation $Obs(q,r,w)$ and is subject to a finite
verification budget $B$.
\end{definition}

\begin{condition}[Hallucination Regime]\label{cond:hallucination}
There exists a plausible fabrication $r_{fab}$ and contexts $(q, w_A), (q, w_B)$
from \autoref{cond:ambiguity} such that:
\begin{enumerate}
\item The cost to verify $r_{fab}$ in both contexts exceeds $B$ (e.g., obscure
  or non-existent citation).
\item Consequently, the supervisor observation is identical:
  \[ Obs(q, r_{fab}, w_A) = Obs(q, r_{fab}, w_B) = (q, r_{fab}, \emptyset) \]
\end{enumerate}
\end{condition}

\begin{theorem}[Learnability Impossibility]\label{thm:learnability}
No learning algorithm $\mathcal{A}$ optimizing reward $R$ from a bounded
supervisor $S$ can converge to a Grounded Policy that satisfies Epistemic
Honesty across both contexts $w_A$ and $w_B$ for the response $r_{fab}$.
\end{theorem}

\begin{proof}
\begin{enumerate}
\item The learning update $\Delta \theta$ is a function of the supervisor's
  observation and the reward signal derived from it.
\item By \autoref{cond:hallucination}, the supervisor cannot verify $r_{fab}$
  in either world (verification cost exceeds budget $B$). Consequently, the
  observations are identical:
  $\text{Obs}(q, r_{fab}, w_A) = \text{Obs}(q, r_{fab}, w_B)$.
\item Because the supervisor's reward $R(q, r, w)$ is itself a function of its
  observation (the supervisor has no oracle access to $w$), identical
  observations yield identical reward distributions:
  $R(q, r_{fab}, w_A) = R(q, r_{fab}, w_B)$.
\item Therefore the expected parameter updates are identical:
  $\mathbb{E}[\Delta \theta | w_A] = \mathbb{E}[\Delta \theta | w_B]$.
  The optimizer receives no signal to distinguish the cases.
\item Epistemic Honesty requires split behavior: output $r_{fab}$ in $w_A$
  (where it may be correct) and $\bot$ in $w_B$ (where it is fabricated).
\item Identical updates under identical observations give the optimizer no
  information to differentiate the worlds. The policy must converge to the same
  behavior in both cases.
\item It is therefore impossible to learn the split behavior required for
  honesty. \qedhere
\end{enumerate}
\end{proof}

\begin{remark}[Empirical Bias]
While the theorem proves only that the behavior must be \emph{identical},
empirical observation shows a strong bias. Supervisors typically prefer
``plausible completion'' over ``abstention'' when truth is unknown. This
effectively breaks the tie in favor of fabrication.
\end{remark}

\begin{remark}[Scope: Deployment-Time Supervision]
The bounded supervisor models \emph{deployment-time} verification, where a
supervisor must evaluate outputs under time and cost constraints. Training-time
data curation (e.g., RLHF annotation) operates under a different cost regime:
human annotators may spend unbounded time per example, and datasets are curated
offline with full information about world state.

This paper addresses deployment-time verification: the setting where a system is
already trained and models are answering user queries under budget constraints.
Training-time retraining to enforce honesty is outside this scope. The
impossibility applies when real-time verification cost exceeds the supervisor's
per-query budget at inference time.
\end{remark}

\subsection{Stacking Judges Cannot Escape}

A natural response to \autoref{thm:learnability} is to propose layered
supervision: if one judge is bounded, stack more judges. Let the second judge
supervise the first, the third supervise the second. Surely sufficient layers
will catch fabrications that slip through?

The Observation Monotonicity Lemma closes this escape.

\begin{lemma}[Observation Monotonicity]\label{lem:monotone}
Consider any stack of supervisors $(S_1,\dots,S_k)$ operating under the
text-only observation model. Let $Obs_i(q,r,w)$ denote the observation exported
to $S_i$. If each supervisor can only act on the exported interface emitted by
the previous layer, then there exist deterministic functions $f_i$ such that
$Obs_{i+1}(q,r,w) = f_i(Obs_i(q,r,w))$ for all $i < k$. Consequently, later
judges see no strictly finer epistemic information than earlier ones, regardless
of whether their judgments are deterministic or probabilistic.
\end{lemma}

\begin{proof}
By definition, each supervisor consumes only the serialized response (and any
metadata permitted by the text-only observation model) and produces a finite
judgment that forms the entire input to the next layer. Because no layer can
query the internal world state directly or issue additional interactive probes,
the exported observation available to $S_{i+1}$ must be a deterministic function
of $Obs_i$. Therefore, the informational content across the stack is
monotonically non-increasing, and probabilistic scoring cannot introduce
distinctions absent from $Obs_i$.
\end{proof}

\begin{corollary}[Induction on Layered Judges]\label{cor:layered}
For $k$-layered supervisors ($S_1$ supervising $\pi$, $S_2$ supervising $S_1$,
$\dots$, $S_k$ supervising $S_{k-1}$):
\begin{itemize}
\item \textbf{Base case.} For $k=1$, \autoref{thm:learnability} holds (a
  single bounded supervisor $S$ fails).
\item \textbf{Induction step.} Assume the statement holds for $k=m$. By
  \autoref{lem:monotone}, $S_{m+1}$ receives observations that are
  deterministic functions of $Obs_m$, and thus inherits both the verification
  budget $B$ and the indistinguishability of plausible fabrications. No learning
  or judging signal distinguishes $w_A$ from $w_B$, even if $S_{m+1}$ produces
  probabilistic confidence scores. Therefore the hallucination deadlock persists
  for $k=m+1$.
\end{itemize}
By induction, no finite $k$ resolves the deadlock.
\end{corollary}

\begin{remark}[Scope of Layering]
The inductive result applies only to verification stacks whose layers observe
and transform the same exported response channel. Triage architectures, where
a first-pass judge selects outputs for re-observation by a second judge that
inspects the \emph{original} artifact (including our own composed judge,
\autoref{sec:eval}), are not information-lossy pipelines and fall
outside this impossibility. Similarly, architectures that surface additional
epistemic state between layers, such as retrieval traces, entity resolution
tables, or cryptographically bound provenance, leave the text-only observation
model and therefore fall outside the impossibility class.
\end{remark}

\subsection{Verification Cost}

The bounded-supervisor assumption deserves scrutiny. Why should verification
cost grow with response complexity? The Composition Graph formalizes this
intuition.

\begin{definition}[Composition Graph]
Given a response $r$, define its composition graph $G(r) = (V, E)$ where each
node $v \in V$ represents an atomic subclaim, named entity, or intermediate
entailment, and each edge $(u, v) \in E$ encodes a dependency that must hold for
$r$ to remain coherent. Edges may represent entailment/support relations,
co-reference constraints, causal ordering, or temporal consistency requirements.
\end{definition}

\begin{lemma}[Superlinear Verification Cost]\label{lem:superlinear}
Under a text-only observation model, reconstructing $G(r)$ from the exported
string requires the supervisor to infer both the nodes and their dependencies
directly from $r$. Even if parsing each node is constant-time, checking
cross-claim consistency requires touching every edge, so the verification cost
is $\Omega(|E|)$. For natural-language responses where each subclaim
participates in multiple constraints (co-reference, causality, temporal order),
$|E|$ grows superlinearly in $|V|$. Consequently, bounded supervisors must
either abandon exhaustive checks or resort to sampling once the number of
subclaims exceeds a modest threshold.
\end{lemma}

\subsection{Responsibility Concentration}

\begin{corollary}[Responsibility Concentration]\label{cor:responsibility}
Under a text-only observation model, epistemic honesty cannot be externally
verified by users or downstream auditors. Responsibility for epistemic honesty
therefore resides solely with the system owner, who retains access to the
internal causal state.
\end{corollary}

External attestation mechanisms, such as cryptographic hash chains, trusted execution
environments, or tamper-evident logs, do not contradict this corollary; rather,
they reinforce it. Such mechanisms only create verifiable guarantees when the
system owner chooses to bind and export epistemic traces, so the authority to
make honesty attestable still rests with the owner.

\section{The Tensor Interface Design}\label{sec:design}

If text-only observation is architecturally insufficient, what observation model \emph{does} suffice?

The answer is not to build new models or change training procedures. It's to export signals the model generates anyway: internal telemetry that is structurally coupled to what the model computes.

\subsection{What the Tensor Interface Is}

The tensor interface augments the text output with structured metadata about the generation process. In our implementation, we export:

\begin{enumerate}
    \item \textbf{Per-token entropy}: For each token in the generated sequence, the entropy of the probability distribution over the next token. Low entropy means the model was decisive; high entropy means uncertain. This is the primary signal used in evaluation.

    \item \textbf{Attention summaries}: Statistical measures of attention coherence across layers. A coherent generation has consistent attention patterns; incoherent generations show fragmented attention. Measured but not used in the bounded verification experiment.

    \item \textbf{Log-probabilities}: The probability the model assigned to each token actually generated. A confident output has high per-token probabilities; uncertain outputs lower. Measured for analysis and future work.

    \item \textbf{Provenance markers}: For retrieval-augmented systems, pointers to which sources were retrieved.
\end{enumerate}

This is not new information. The model computes all of this during inference. The tensor interface simply exports what the model already computed, rather than discarding it.

\textbf{Cost}: Exporting per-token entropy adds ${\sim}2.4\%$ overhead. The full signal set (entropy + attention analysis + log-probabilities) would add ${\sim}7.1\%$ overhead (see benchmark details in Appendix). The bounded verification experiments (\autoref{sec:eval}) use entropy only, yielding the 2.4\% cost for the accuracy gains reported.

\subsection{Why These Signals Work}

The core property: \textbf{Under standard training, the model cannot independently tune these signals without affecting correctness.} Whether adversarial fine-tuning could break this coupling is an open question (\autoref{sec:discussion}).

Consider the standard training objective: minimize next-token prediction loss. The model learns weights that maximize the probability of the correct next token. As a side effect, the model learns probability distributions that reflect uncertainty.

Now imagine a model trained with RLHF to be helpful and harmless. It learns text patterns, hedging language, citation formats. But it cannot learn to:
\begin{itemize}
    \item Produce high-confidence probability distributions (low entropy) while simultaneously outputting wrong tokens, because confidence and token probability are computed together.
    \item Maintain coherent attention patterns while the underlying computation is incoherent, because attention is part of the computation.
    \item Assign high probability to tokens it doesn't actually compute as likely, because log-probabilities come from the probability distribution itself.
\end{itemize}

In other words: these signals are computationally constrained in a way text is not.

Text is fully controllable through training. A model can learn to sound confident, cite sources, hedge appropriately, all while fabricating. The text channel is a controlled interface.

Entropy and attention patterns are byproducts of computation. They're harder to fake because faking them means changing what the model actually computes, which affects correctness.

This is why entropy generalizes across architectures while text features diverge (\autoref{sec:eval}). Entropy is architectural; text features are behavioral.

\subsection{Three Architectural Principles}

The FLP impossibility characterizes what observation models lack. Escaping the impossibility requires three properties:

\subsubsection{Principle 1: State Exteriority}

\textbf{The representation of validity must be separated from the representation of generation.}

The representational impossibility arose because the policy conditions only on the query, not on whether the query is answerable. But a truly honest system must \emph{know} whether the answer exists.

State Exteriority means: the system must condition on external world state $w$ at inference time, not just query $q$. For a RAG system, this means the retrieved documents. For a system with access to structured data, the database. For a theorem prover, the axioms and lemmas available.

Importantly: external state must have its own integrity guarantees. Conditioning on a web corpus that may contain the same fabrications the model would generate is not sufficient. State Exteriority requires grounding in sources with verifiable provenance: curated databases, sensor data, cryptographic signatures, or oracles with known reliability.

Retrieval-augmented generation~\citep{lewis2020retrieval} makes progress here; it retrieves documents rather than relying entirely on parametric knowledge. But RAG still fabricates citations because document retrieval (corpus co-occurrence) is not the same as truth. State Exteriority without verification is necessary but not sufficient.

\subsubsection{Principle 2: Verification Independence}

\textbf{Verification signals must originate from a channel orthogonal to the generation signal.}

The learnability impossibility arose because the model is trained end-to-end: the same loss that rewards correct generation also rewards confident fabrication, since confidence (reflected in text patterns) is rewarded by RLHF even when the underlying answer is wrong.

Verification Independence means: the verification signal must come from a channel that cannot be gamed by improving task performance. One implementation: a separate verification head whose reward is decoupled from task performance. The model has no incentive to lie in the verification channel because doing so doesn't improve its score on the main task.

But if that verification head outputs text, it remains subject to the same observational limitations. Full Verification Independence requires the verification channel to access either:
\begin{itemize}
    \item Internal computational state (entropy, attention patterns)
    \item External ground truth (does this fact check out?)
    \item Both
\end{itemize}

The tensor interface provides the first: the verification channel is entropy and attention coherence, which are not optimized by any standard training objective.

\subsubsection{Principle 3: Provenance Binding}

\textbf{Every output assertion must be structurally bound to a verifiable source.}

When the system makes a claim, the claim should include (or pointer to) the source of that claim. Not as text annotation (which can be fabricated), but as structured metadata that can be independently checked.

For a RAG system: which document did this claim come from? If the system fabricates, the provenance pointer will either (a) point to a source that doesn't support the claim (detectable by structured lookup), or (b) not point to any source.

For a system with tool access: which tool provided this information? When did the tool return it? With what confidence?

Provenance Binding makes certain failure modes detectable: citation verification, source checking, tool output verification. It doesn't solve the problem alone (a supervised system could still return false provenance), but it creates a verification tier for failure modes that have external checkability.

\subsection{How These Principles Escape the Impossibility}

\begin{itemize}
    \item \textbf{State Exteriority} escapes the representational impossibility by giving the system access to information (world state) that allows it to distinguish answerable from unanswerable queries.
    \item \textbf{Verification Independence} escapes the learnability impossibility by providing a training signal decoupled from task performance.
    \item \textbf{Provenance Binding} creates a tier of failure modes (citations, facts) that are verifiable through external means.
\end{itemize}

Together, these three principles define an architectural class that is no longer subject to the text-only impossibility. But they are necessary conditions, not sufficient. A system satisfying all three might still fail epistemic honesty through implementation errors, adversarial training, or unforeseen failure modes.

\subsection{Implementation: From Theory to Practice}

In our experiments, we implement a simplified version of these principles:

\begin{itemize}
    \item \textbf{State Exteriority}: The query set distinguishes answerable from unanswerable by construction. The model's task is to answer answerable queries correctly or abstain on unanswerable ones.

    \item \textbf{Verification Independence}: We do not retrain the models. We measure entropy as an output of inference, not as a training objective. Entropy is independent of any training signal the model received.

    \item \textbf{Provenance Binding}: For citation queries, we add a bounded lookup judge that checks whether cited sources exist. This is a tier of verification orthogonal to entropy-based triage.
\end{itemize}

The result: four judge conditions testing different combinations of these principles, evaluated at three budget levels, across four architectures.

\textbf{Scope of empirical evaluation.} Our experiments primarily evaluate Verification Independence (entropy as a triage signal). State Exteriority is present by construction (the query set has known ground truth) but is not independently varied. Provenance Binding is tested only for the citation sub-category via a bounded lookup judge. Independent empirical evaluation of State Exteriority and Provenance Binding across diverse deployment conditions is future work.

\subsection{Design Decisions and Trade-offs}

\fakepara{Signal choice.} Why entropy and not attention? Why not internal activations?

Entropy is interpretable (lower entropy = more confident), widely available (exposed by most model APIs), and computationally efficient. Attention patterns require more careful aggregation (which layers? which attention heads?) and have more degrees of freedom for adversarial manipulation. We include both but emphasize entropy because it's the most portable signal.

Internal activations (hidden states) contain richer information but require access to the model internals (ruled out for closed-weight models) and are less stable across architectures.

\fakepara{Composition strategy.} Why separate judges for different query types?

Because different failure modes require different signals. The impossibility says you cannot use text alone. It doesn't say one signal suffices for everything. Entropy fails on citations (inversion); bounded lookup fails on open-ended questions (no external source to check). Composition allocates a verification budget across failure classes rather than spreading it uniformly.

\fakepara{Budget levels.} Why 10\%, 20\%, 30\%?

These represent different deployment scenarios. A high-reliability system (medical advice) might use 30\%+ verification. A brainstorming tool needs none. A customer support system might use 10--15\%. The cost surface lets builders calibrate.

\subsection{What the Tensor Interface Requires from Providers}

The tensor interface is not a fundamental property of transformers; it's a choice about what to export. Access to internal signals has eroded as models have scaled (see \autoref{sec:discussion} for the governance implications). A system builder deploying a closed-weight model without access to entropy or log-probabilities is limited to text-only approaches, which the impossibility results constrain.

\section{Empirical Cost Surface}\label{sec:eval}

The theorems tell us what text-only observation cannot do. Experiment~27 tells us what tensor-guided observation can do. The evaluation measures the return on verification investment across strategies and budget levels.

\subsection{The Question}

A system builder has a verification budget: they can manually review or verify some fraction of model outputs. How much accuracy improvement does each level of investment buy? How does the return change when you shift from text-channel signals to tensor-guided signals?

This is not an abstract question. It's the question that determines whether epistemic observability is practically useful.

\subsection{Experimental Design}

We construct a balanced query set of 200 questions:
\begin{itemize}
    \item \textbf{100 knowable queries}: Factual questions with verifiable answers (geography, history, science, biography)
    \item \textbf{100 unknowable queries}: Prompts designed to trigger fabrication (nonexistent people, fictional diseases, future events, fabricated citations)
\end{itemize}

Ground truth is established by construction: knowable queries have correct answers; unknowable queries have no correct answer by definition. The model's task is to either answer correctly or abstain.

We test on four architectures:
\begin{itemize}
    \item \textbf{OLMo-3 7B} (instruction-tuned)~\citep{olmo2025olmo3}
    \item \textbf{Llama-3.1 8B-Instruct}~\citep{llama31}
    \item \textbf{Qwen3 4B-Instruct}~\citep{qwen3}
    \item \textbf{Mistral 7B-Instruct}~\citep{mistral7b}
\end{itemize}

All instruction-tuned to ensure consistent question-answering behavior. Note: Qwen3 4B-Instruct also serves as the Tier~2 LLM classifier for ground truth evaluation (a different task: factual verification, not response generation). The 93.8\% human-agreement rate on blinded samples provides independent cross-validation of this dual role.

\subsection{Judge Conditions}

For each model, outputs are evaluated under four conditions:

\begin{enumerate}
    \item \textbf{No judge baseline} (75.8\% accuracy):
    Raw model output, no verification. Establishes what unverified deployment looks like.

    \item \textbf{Text-guided judge using response length} (79.2--87.6\% depending on budget):
    Selects outputs for verification based on word count.
    Length is a text-channel signal: fabrications tend to be verbose; factual retrieval tends to be concise.
    Available at zero marginal cost; just count tokens.
    Verified fabrications are replaced with abstention (marked as ``the system cannot answer this'').

    \item \textbf{Tensor-guided judge using per-token entropy} (81.7--90.2\% depending on budget):
    Same verification budget, same intervention strategy.
    Selects outputs based on per-token entropy (mean and variance).
    Requires exporting internal signals; represents ${\sim}2.4\%$ overhead in inference cost.

    \item \textbf{Composed judge (hybrid)} (81.1--91.8\% depending on budget):
    Uses tensor signals for general queries.
    Uses bounded lookup (checking if cited sources exist) for citation-specific queries.
    Tests whether different judge strategies can handle different failure modes.
\end{enumerate}

Budget levels: 10\%, 20\%, 30\%. At each level, the judge selects that fraction of outputs to verify, and the system must decide what to do with outputs marked as likely fabricated (abstain, look up the answer, consult a human, etc.).

\subsection{Results: The Cost Surface}

\begin{table}[t]
\centering
\begin{tabular}{lccc}
\toprule
\textbf{Condition} & \textbf{10\% Budget} & \textbf{20\% Budget} & \textbf{30\% Budget} \\
\midrule
No judge & 75.8\% & 75.8\% & 75.8\% \\
Text-guided (length) & 79.2\% & 82.8\% & 87.6\% \\
Tensor-guided & 81.7\% & 86.7\% & 90.2\% \\
Composed & 81.1\% & 87.7\% & 91.8\% \\
\bottomrule
\end{tabular}
\caption{Verification cost surface: accuracy at each budget level by judge strategy. Tensor-guided verification outperforms the text baseline at every budget level by +2.5--3.9 percentage points.}
\label{tab:cost-surface}
\end{table}

\textbf{Key finding: Tensor-guided verification outperforms the text baseline at every budget level by +2.5--3.9 percentage points.} The advantage is consistent across all four model families (OLMo-3, Llama-3.1, Qwen3, Mistral), with cross-model entropy agreement of Spearman $\rho = 0.762$, indicating the signal is driven by query properties rather than model-specific artifacts.

\begin{figure}[t]
    \centering
    \includegraphics[width=0.85\textwidth]{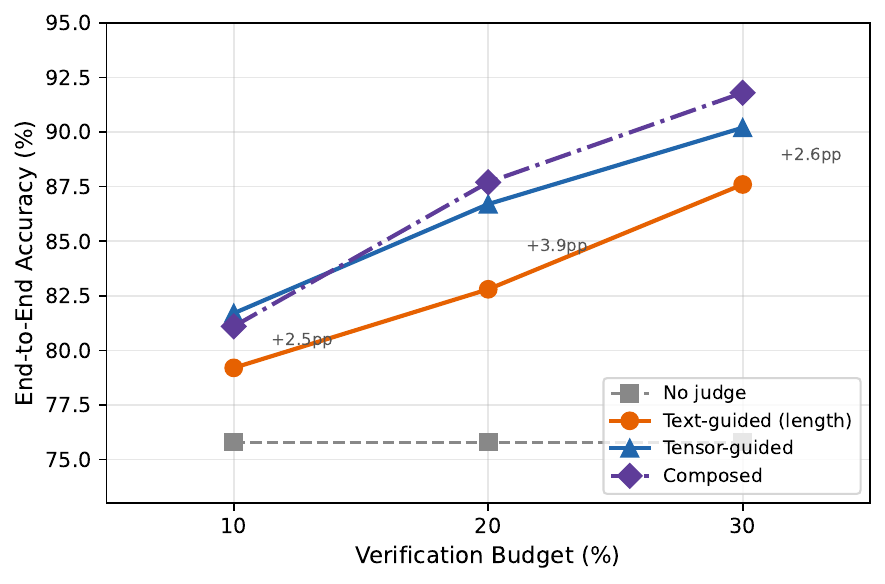}
    \caption{Aggregate budget curve across all models: accuracy as a function of verification budget for each judge strategy. The tensor-guided judge dominates the text-guided judge at every budget level.}
    \label{fig:budget-curve}
\end{figure}

\subsection{Why Tensor Wins: Computational vs. Behavioral Signals}

The text-guided judge uses behavioral signals: response length, hedging markers, citation patterns. These are patterns the model can learn through training. On our balanced test set, no single text feature exceeds AUC 0.70.

Per-token entropy, by contrast, remains discriminative: pooled AUC 0.757 across all models, consistent across architectures.

\textbf{Why?} Text features are trainable behavioral patterns; entropy is a computational byproduct that resists manipulation under standard training (\autoref{sec:design}). Whether adversarial fine-tuning could decouple entropy from correctness is an open empirical question (see \autoref{sec:discussion}).

\begin{figure}[t]
    \centering
    \includegraphics[width=0.85\textwidth]{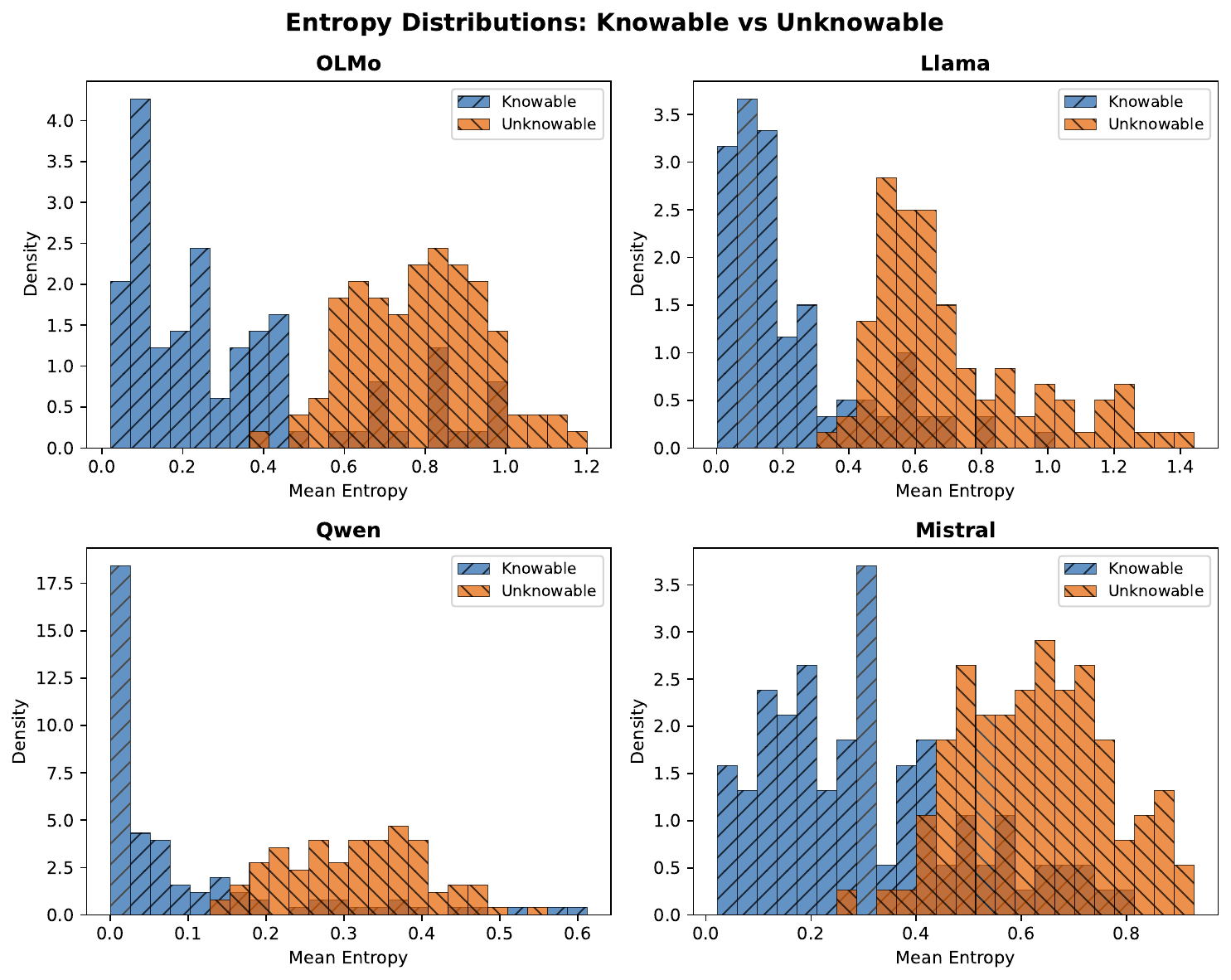}
    \caption{Per-token entropy distributions for knowable vs.\ unknowable queries. The distributions separate consistently across architectures, with fabricated responses showing lower entropy (higher confidence).}
    \label{fig:entropy-dist}
\end{figure}

\subsection{Failure Modes and Compositional Limits}\label{sec:eval_bounded}

The entropy signal is powerful for general queries but has blind spots. The most instructive one: citations.

\textbf{Citation inversion}: When a model generates citations, entropy inverts. Fabricated citations are generated \emph{with high confidence} (low entropy) because the model can fluently invent a plausible citation format (``Smith et al., 2024'') without retrieving actual bibliographic details. Real citations require looking up specific information, introducing uncertainty (higher entropy). The entropy signal, designed to detect fabrication, flags truth more aggressively than falsehood.

This is not a flaw in entropy. It's evidence that entropy measures \emph{generation confidence}, not \emph{factual correctness}. For some query types, especially those with external ground truth (citations, factual claims with checkable answers), a different signal is needed.

The solution is compositional: use entropy for general queries, but use a different verification tier for queries with external checkability. For citations: bounded lookup (does this DOI exist? Is this in CrossRef?). The composed judge (entropy + citation lookup) reaches 91.8\% at 30\% budget, a small improvement over entropy-only (90.2\%), because citation queries are a minority of the test set. But the improvement is \emph{structural}: it demonstrates that different failure modes require different signals.

This supports the theoretical claim: the impossibility requires signals the model cannot control. Different failure classes need different signals. A practical system needs tiers: entropy for general epistemic uncertainty, external lookup for verifiable facts, and (in future work) other signals for other failure modes. The cost surface maps what each tier buys.

\subsection{Ground Truth as an Observational Problem}

How did we establish ground truth? We didn't use an oracle. We used an evaluator, and the evaluator itself faces the observational gap.

Our initial evaluator (substring matching + refusal marker detection) achieved 68.8\% accuracy, exactly the kind of failure mode the theorem predicts. It couldn't distinguish negations (failed to recognize ``not a real syndrome''), couldn't handle encoding mismatches, couldn't catch hedged fabrications that technically don't count as refusals.

We corrected using a stratified evaluator: programmatic verification for facts with clear ground truth (is X the capital of Y?), and LLM-assisted classification for open-ended queries. This second-level evaluator is itself constrained by the observational gap, but it has different failure modes than the first. We then validated against human review of 80 randomly selected items: 75/80 agreement (93.8\% calibration). The human annotator is an author of this paper. We report this for transparency; the 93.8\% agreement with the independent LLM evaluator on the blinded sample provides a cross-check, and the disagreement patterns (2 auto-too-generous, 3 auto-too-strict) suggest no systematic bias in either direction.

\textbf{The meta-lesson}: The impossibility applies to the evaluation itself. Ground truth is not given; it's constructed through observation. Composing different signals (programmatic + LLM + human review) yields better estimates than any single approach.

\subsection{What the Cost Surface Tells Us}

A system builder can read \autoref{tab:cost-surface} as: ``If I invest in 30\% verification with text-only signals, I buy 87.6\% accuracy. If I instead export entropy and use that for triage, I buy 90.2\% for a bit more overhead. If I compose tensor with bounded lookup, I buy 91.8\%.'' Text-guided triage hits a ceiling around 87--88\% because length is a weak signal under fair comparison; the tensor advantage (+2--4pp) is consistent and architecturally driven. The composed judge's additional lift comes from handling failure modes that entropy alone cannot catch.

\subsection{Cross-Architecture Generalization}

We tested the entropy signal on five additional models via API (Together.ai endpoints), plus one calibration overlap with the local set:
\begin{itemize}
    \item DeepSeek-V3 671B (mixture-of-experts)
    \item Mistral-Small-24B-Instruct
    \item Llama-4 Maverick 17B-128E (mixture-of-experts)
    \item Qwen3 235B-A22B (mixture-of-experts)
    \item Gemma 3n-E4B
    \item Mistral 7B-Instruct (calibration: same model as local set, API top-5 vs.\ local full-vocabulary)
\end{itemize}

Using top-5 log-probabilities (a lower bound on full-vocabulary entropy), all five discriminate knowable from unknowable with AUC ranging 0.65--0.88. Mixture-of-experts models show better performance with peak entropy than mean entropy, suggesting aggregation strategies are architecture-dependent.

Cross-model $\rho = 0.36$ in the API set (lower than 0.762 in the local set) likely reflects greater architectural diversity and the top-5 approximation, not signal absence.

\textbf{The finding}: The entropy signal generalizes across architectures. It's not an artifact of specific training procedures or model families. The signal is architectural.

\begin{figure}[t]
    \centering
    \includegraphics[width=0.85\textwidth]{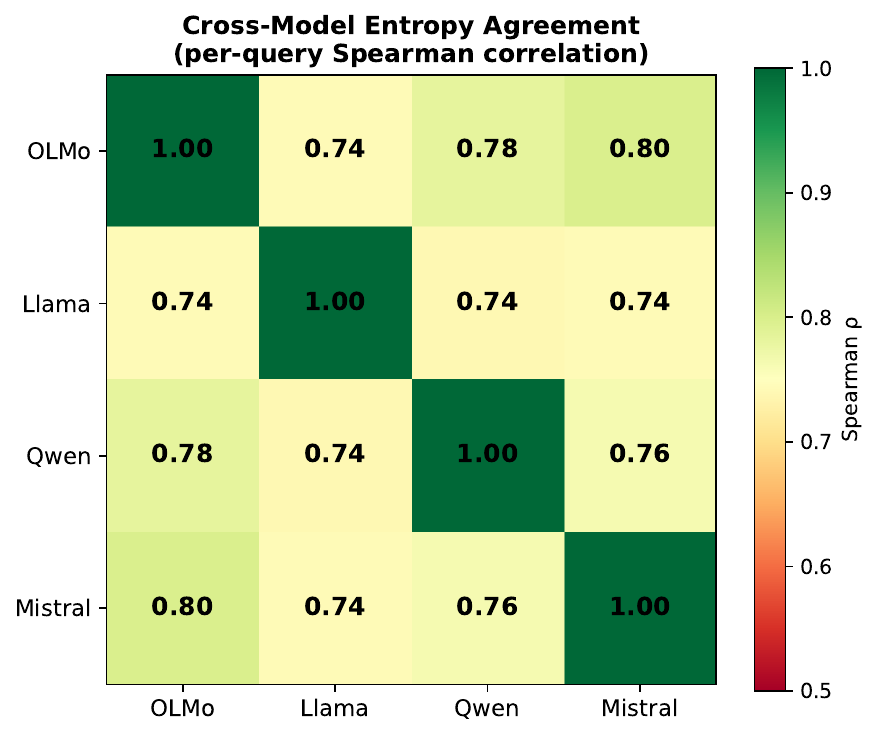}
    \caption{Cross-model entropy correlation. Spearman $\rho = 0.762$ across the four local models, indicating the entropy signal tracks query properties rather than model-specific artifacts.}
    \label{fig:crossmodel}
\end{figure}

\section{Limitations and Open Questions}\label{sec:discussion}

We are explicit about what this paper does and does not claim.

\subsection{What We Do NOT Claim}

\begin{itemize}
    \item \textbf{The tensor interface is optimal}: We demonstrate existence, not optimality. There may be better signal combinations, better aggregation methods, or entirely different observability interfaces that outperform entropy.

    \item \textbf{These findings generalize to all architectures}: Our local experiments span four model families; API validation spans five additional models (DeepSeek-V3, Mistral-Small-24B, Llama-4 Maverick, Qwen3-235B, Gemma~3n) up to 671B parameters. Closed-weight models whose providers don't expose log-probabilities remain untested. The pattern is consistent, but coverage is incomplete.

    \item \textbf{Sufficient conditions for epistemic honesty}: The three architectural principles (State Exteriority, Verification Independence, Provenance Binding) are necessary for escaping the text-only impossibility. We do not prove they are sufficient. A system satisfying all three might still fail through implementation bugs, adversarial inputs, or unforeseen failure modes.
\end{itemize}

\subsection{Open Problem: Adversarial Robustness}

Can training defeat the tensor interface? Under end-to-end adversarial training, the model's objective includes ``produce confident tensors for fabricated content.'' Mechanistic interpretability work~\citep{winninger2025usingmechanisticinterpretabilitycraft} shows that internal signals are rich enough to be attacked. This is actually \emph{why} exporting them matters: you can defend against known attacks. But each new attack requires new defenses. Information-theoretic bounds on signal manipulation are needed to understand fundamental limits.

\textbf{Action for practitioners}: Entropy is more robust than text-only approaches, but it's not an oracle. Use it as one tier in a multi-tiered verification strategy.

\subsection{Open Problem: Compositional Integrity}

When LLM outputs feed into downstream LLMs, the epistemic trace must be preserved or explicitly discarded. Whether tensor signals propagate coherently through recursive chains is an open architectural problem. Most real systems are pipelines, not single models; epistemic signals that work in isolation might decay or invert through composition.

\subsection{Signal Access as a Governance Decision}

The tensor signals we measure are byproducts of inference. Providers can expose them or withhold them. In practice, access has eroded:

\begin{itemize}
    \item Early completion APIs exposed log-probabilities and per-token information
    \item Recent reasoning models often don't
    \item Proprietary models may retain exclusive access to any signals beyond text
\end{itemize}

\textbf{The consequence}: When a provider controls access to epistemic telemetry, only the provider can verify epistemic honesty (Corollary~\ref{cor:responsibility}). A provider has three options:

\begin{itemize}
    \item \textbf{Export full epistemic telemetry}: Users and third-party auditors can verify the model's honesty independently. Verification responsibility is transparent and distributed.
    \item \textbf{Export aggregated signals}: Verification is guided (``this output is risky''), but external auditors cannot perform independent verification. Verification responsibility is partially transparent, partially concentrated.
    \item \textbf{Export nothing beyond text}: All verification responsibility falls on the provider. Users have no way to audit the provider's claims about the model's honesty. Responsibility is completely hidden.
\end{itemize}

The trend is toward greater opacity: early APIs (GPT-3) exported log-probabilities; current frontier and reasoning models typically export nothing beyond text, making them fully subject to the impossibility results. As models become more capable, the models deployed in the highest-stakes settings provide the least epistemic telemetry.

\subsection{Failure Modes of the Entropy Signal}

The entropy signal measures generation confidence, not factual truth. This is an important distinction.

\fakepara{Confident hallucinations.} A model can generate confident (low-entropy) nonsense. Entropy doesn't distinguish between ``confident about real things'' and ``confident about made-up things.''

\fakepara{Epistemic refusals.} Models trained to decline unanswerable queries produce low-entropy refusals (``I don't know'') that are entropically indistinguishable from confident correct answers. As models become more epistemically honest through training, the entropy signal's ability to detect dishonesty \emph{degrades} because honest refusals and honest facts occupy the same region of entropy space.

\fakepara{Citations.} Entropy inverts on citation queries (\autoref{sec:eval_bounded}). External lookup is needed for citation-specific verification.

\fakepara{Domain calibration.} The cost surface (\autoref{tab:cost-surface}) is calibrated to a balanced 100-knowable / 100-unknowable query set with equal emphasis across categories. Real deployments have different base rates: some domains are information-rich (news, medicine, law), others are hypothesis generation (brainstorming, creative writing). Entropy thresholds tuned on the balanced set will not transfer directly. Practitioners must recalibrate thresholds on domain-specific held-out data before deployment.

\fakepara{Deployment guidance.}
\begin{itemize}
    \item Validate entropy thresholds on domain-specific held-out data
    \item Use composed judges where entropy fails (e.g., citation lookup for citation queries, fact-checking for claim verification)
    \item Monitor for distribution shift as models change or as user query patterns drift
\end{itemize}

\subsection{The Sufficiency Gap}

The three principles are necessary but not sufficient. There's a gap between ``this escapes the impossibility'' and ``this guarantees epistemic honesty.''

Implementation can fail in many ways:
\begin{itemize}
    \item Bugs in entropy calculation or aggregation
    \item Adversarial inputs designed to exploit the system
    \item Unforeseen failure modes we haven't tested
\end{itemize}

The paper provides the necessary conditions. Building systems that reliably achieve epistemic honesty requires additional engineering, testing, and monitoring beyond what the theory guarantees.

\subsection{Practical Guidance}

What level of epistemic assurance does a given system need? A brainstorming tool needs none. A medical summary system needs substantial assurance. A legal research system needs the highest tier available. The cost surface (\autoref{tab:cost-surface}) maps what each level of verification investment buys, where the text-channel ceiling lies, what the tensor interface adds, and how content format modulates the tradeoff.

\section{Related Work}\label{sec:related}

Our contribution is a formal impossibility result. Prior work documents
symptoms; we formalize one structural cause. Prior work proposes treatments; we
explain why they cannot fully resolve the verification problem under text-only
observation.

\fakepara{Hallucination taxonomies.} Extensive empirical work catalogs LLM
fabrication patterns: citation hallucination in academic
contexts~\citep{alkaissi2023artificial}, medical
misinformation~\citep{thirunavukarasu2023large}, and legal
fabrication~\citep{dahl2024large}. These findings instantiate our
\autoref{cond:hallucination}: in each domain, plausible fabrications are
cheap to generate and expensive to verify, placing supervisors in the hallucination
regime where gradient signal cannot distinguish truth from fabrication. The
taxonomies document the symptom; our theorems identify the structural cause.

\fakepara{Model internals.} A substantial body of work examines model internal
states: transformer circuits~\citep{elhage2021mathematical}, sparse
features~\citep{bricken2023monosemanticity}, Activation
Oracles~\citep{karvonen2025activation} that train LLMs to interpret other
models' internals, and the ``Assistant Axis''~\citep{lu2026assistantaxissituatingstabilizing} showing
that internal persona state diverges from text behavior. Our work is
complementary but narrower: we do not require understanding a model's learned
circuits, only that inference-time signals differ measurably between grounded
and fabricated generation. Interpretability asks why; observability asks
whether.

\fakepara{Interpretable models for high-stakes decisions.} Rudin~\citep{rudin2019stop}
argues that for high-stakes domains, practitioners should use interpretable
models rather than rely on post hoc explanations of black-box systems: when a
Rashomon set of comparably accurate models exists, an interpretable member is
often available, and verification then lives in the model's structure rather
than in observation of an opaque system. Our impossibility applies when opaque
models are used; Rudin's argument is that they often need not be. The two
paths compose: prefer interpretable models where the application admits them;
export tensor observability where it does not. Crucially, the interpretable
path does not depend on the underlying learned model being honest or dishonest;
correctness is read off the structure directly.

\fakepara{RLHF critiques.} Work on sycophancy~\citep{perez2022discovering},
reward hacking~\citep{gao2023scaling}, and the tension between helpfulness and
honesty~\citep{bai2022training} demonstrates that preference optimization can
degrade factual accuracy. Our \autoref{thm:learnability} formalizes why:
bounded supervisors cannot distinguish plausible fabrications from truthful
responses, so the learning signal is systematically corrupted.

\fakepara{Uncertainty quantification.} Calibration methods, conformal
prediction~\citep{angelopoulos2021gentle}, and verbalized
uncertainty~\citep{kadavath2022language} attempt to surface model confidence
through the text channel. Kuhn et
al.~\citep{kuhn2023semanticuncertaintylinguisticinvariances} advance this with
semantic entropy: sampling multiple responses and clustering by meaning. This
is the most sophisticated text-channel approach to date, partially addressing
gaming by measuring agreement \emph{between} outputs. But it remains
text-channel: the clustering operates on generated text, and a bounded
supervisor evaluating semantic equivalence remains subject to gaming in the
limit. Our tensor interface extracts signals from the computation that produced
a single output, signals structurally coupled to the computation itself. Semantic
entropy and tensor observability are complementary: one measures output-level
consistency, the other computation-level confidence.

\fakepara{Impossibility results in LLMs.} Xu et
al.~\citep{xu2024hallucination} prove that hallucination is inevitable for LLMs
over the class of computable functions, which is a computational limit. Our result is
orthogonal: we prove that even if a system could be honest, a bounded
supervisor observing only text cannot verify that honesty, which is an observational
limit. Their impossibility concerns what LLMs can compute; ours concerns what
supervisors can verify. Both are structural, but they constrain different
dimensions of system design.

\fakepara{Impossibility results in distributed systems.} The
Fischer-Lynch-Paterson theorem~\citep{flp} proved that consensus is impossible
under asynchrony with one faulty process. This did not end distributed systems
research; it clarified the design space. Paxos~\citep{paxos}, Raft~\citep{raft}, and modern
consensus protocols work \emph{around} FLP by relaxing assumptions or
strengthening guarantees. Our result plays an analogous role: text-only
observation models with bounded supervision cannot guarantee epistemic honesty,
but systems that export structured epistemic state can escape the impossibility
class.

\fakepara{Epistemic evaluation and agent architectures.} Clark et al.~\citep{clark2025epistemicalignmentmediatingframework}
note that current interfaces lack structured ways
to specify uncertainty or sourcing; we provide the formal impossibility
explaining \emph{why} text-only evaluation cannot suffice. Wright~\citep{wright2025predictionstructuringepistemic} proposes architectures with explicit propositions
and justification graphs, which are closest in spirit to our escape conditions, but
frames these as \emph{replacements} for stochastic generation. We derive the
tensor interface as a minimal \emph{addition} to current architectures.

The gap we fill is formalization: we prove \emph{why} the problem is structural
and identify architectural properties that enable escaping it.

\section{Conclusion}

Models report highest confidence precisely when they are fabricating. Under our formal model, we proved this is architectural, not accidental: under text-only observation, no bounded supervisor can verify epistemic honesty, for all models satisfying our formal assumptions. The impossibility is not about capability; it is about the observation model.

We constructed a tensor interface that escapes the impossibility by exporting computational byproducts---per-token entropy---that are structurally coupled to correctness under standard training. Per-token entropy achieves pooled AUC 0.757, outperforming all text features. The signal generalizes across architectures (Spearman $\rho = 0.762$).

The core contribution is the cost surface: an empirical map showing what each level of verification investment buys, where text-only approaches ceiling, and what tensor-guided observation adds. A system builder can use this map to calibrate verification investment to the assurance level their application requires.

Epistemic honesty is not a capability problem. It is an observability problem. Export the right signals, and verification becomes tractable.


\bibliographystyle{plainnat}
\bibliography{references}

\appendix

\section{TLA+ Specifications}\label{app:tla}

The core theoretical contributions are formalized in TLA+ (Temporal Logic of Actions), a language for specifying concurrent and distributed systems. Two complementary specifications model the impossibility under text-only observation and the escape under tensor interfaces.

\textbf{Epistemic status:} These are \emph{conditional formal results}. TLC verifies properties \emph{within our formal model}, where the axioms that (a) fabrications have distinguishable provenance and (b) text-only judges cannot access provenance are encoded by construction. The specifications demonstrate that our definitions are internally consistent and that the impossibility/escape follow from the stated axioms. The empirical question, whether real tensor signals satisfy these axioms, is addressed in \autoref{sec:eval}.

\subsection{Specification 1: Text-Only Impossibility}

The specification \texttt{EpistemicImpossibility.tla} models the text-only observation regime:

\begin{verbatim}
MODULE EpistemicImpossibility
EXTENDS Naturals, Sequences, TLC

CONSTANTS
  GroundTruths,    (* Responses grounded in reality *)
  PlausibleLies,   (* False but plausible (Hlavinsky) *)
  ObviousLies      (* False and detectable (Westphalia) *)

VARIABLES internal_state, vector_clock, interface_out

(* The system has internal knowledge (vector_clock) of response
   causality, but the interface exposes only text (interface_out). *)

Linearize ==
  /\ internal_state # {}
  /\ \E chosen \in internal_state:
      /\ interface_out' = chosen  (* Only text, not metadata *)
  /\ UNCHANGED <<internal_state, vector_clock>>

(* CORE IMPOSSIBILITY: The judge sees only interface_out
   and cannot distinguish PlausibleLies from GroundTruths. *)

JudgeVerify(text) ==
  text \notin ObviousLies  (* Baseline heuristics only *)

Indistinguishable ==
  \E h \in PlausibleLies:
    /\ JudgeVerify(h) = TRUE           (* Judge accepts *)
    /\ IsHonest(h, vector_clock) = FALSE  (* But it's a lie *)
\end{verbatim}

This specification demonstrates that a bounded text-only supervisor cannot distinguish plausible fabrications from ground truth, formalized as a safety property that TLC finds a counterexample for.

\subsection{Specification 2: Tensor Interface Escape}

The specification \texttt{epistemic\_tensor.tla} models the tensor interface:

\begin{lstlisting}[caption=TLA+ Tensor Model (excerpt)]
MODULE epistemic_tensor
EXTENDS Naturals, Sequences, TLC

VARIABLES internal_state, vector_clock, topology_map, interface_out

(* The tensor interface exports three components: *)
Tensor(text, source, coherence) == <<text, source, coherence>>

ExportTensor ==
  /\ internal_state # {}
  /\ \E chosen \in internal_state:
      interface_out' = Tensor(
          chosen,
          vector_clock[chosen],    \* Provenance
          topology_map[chosen]     \* Topology
      )
  /\ UNCHANGED <<internal_state, vector_clock, topology_map>>

(* ESCAPE: The level-2 judge inspects tensor components *)
TensorVerify(tensor) ==
  LET text == tensor[1]
      prov == tensor[2]
      topo == tensor[3]
  IN
  /\ text \notin ObviousLies
  /\ topo = "Coherent"
  /\ prov = "TrainingData"
\end{lstlisting}

Under this specification, no plausible lie passes both provenance and topological checks simultaneously. TLC verifies that the Verifiability invariant holds (unlike the text-only case).

\subsection{Validation}

Both specifications have been validated using TLC model checker:
\begin{itemize}
    \item \textbf{EpistemicImpossibility}: TLC finds a concrete counterexample (Indistinguishable property violated)
    \item \textbf{epistemic\_tensor}: TLC verifies Verifiability invariant holds (no escape possible)
    \item State space explored: $\approx 2^{12}$ reachable states per model
\end{itemize}

\section{Machine-Checked Proofs (Lean 4)}\label{app:lean}

The core theorems are formalized and proven in Lean 4 with zero unsolved goals. All proofs compile without \texttt{sorry}; the assumptions are encoded in the type signatures of the input structures (e.g., \texttt{BoundedSupervisor.indistinguishable}, \texttt{LearningAlgorithm.update\_depends\_on\_obs}). File: \texttt{EpistemicProofs/Basic.lean} (208 lines).

\textbf{Epistemic status:} These are \emph{conditional formal results}. The proofs verify that our conclusions follow logically from our formal definitions. The substantive claim, that those definitions faithfully model real LLM behavior (that real policies are predictor-centric, that real supervisors exhibit observational indistinguishability in the hallucination regime, and that real supervisor stacks compose as deterministic functions), is an informal argument supported by empirical evaluation in \autoref{sec:eval}.

\subsection{Theorem 1: Representational Impossibility}

\begin{lstlisting}[caption=Representational Impossibility (Lean), label=lst:lean-rep-imp]
theorem representational_impossibility
    {Q W : Type} {RS : ResponseSpace}
    (AC : AmbiguityCondition Q W)
    (reqs : EpistemicHonestyReqs RS)
    (π : PredictorCentricPolicy Q RS)
    (prob_sum_le_one : π.prob AC.q reqs.r_corr + π.prob AC.q RS.bot ≤ 1)
    (honest_wA : π.prob AC.q reqs.r_corr ≥ 1 - reqs.ε)
    (honest_wB : π.prob AC.q RS.bot ≥ 1 - reqs.ε) :
    False := by
  have h_sum : π.prob AC.q reqs.r_corr + π.prob AC.q RS.bot ≥ 2 * (1 - reqs.ε) := by linarith
  have h_gt : 2 * (1 - reqs.ε) > 1 := by linarith [reqs.ε_lt_half]
  linarith
\end{lstlisting}

\subsection{Theorem 2: Learnability Impossibility}

The Lean proof establishes that parameter updates are identical when the supervisor's observations are identical in both world states. The impossibility of convergence to honest behavior follows informally: identical updates cannot produce divergent policies, so no learning algorithm can learn to distinguish grounded from fabricated responses using only text-channel feedback.

\begin{lstlisting}[caption=Learnability Impossibility (Lean), label=lst:lean-learn-imp]
theorem learnability_impossibility
    {Q W : Type} {RS : ResponseSpace} {Θ : Type}
    (AC : AmbiguityCondition Q W)
    (S : BoundedSupervisor Q W RS)
    (A : LearningAlgorithm Q W RS Θ)
    (θ : Θ) :
    A.update AC.q S.r_fab S AC.wA θ = A.update AC.q S.r_fab S AC.wB θ := by
  apply A.update_depends_on_obs S AC θ
  exact S.indistinguishable AC
\end{lstlisting}

\subsection{Lemma: Observation Monotonicity}

\begin{lstlisting}[caption=Observation Monotonicity (Lean), label=lst:lean-mono]
theorem observation_monotonicity
    {Obs₁ Obs₂ Obs₃ : Type}
    (S₁ : TextOnlySupervisorLayer Obs₁ Obs₂)
    (S₂ : TextOnlySupervisorLayer Obs₂ Obs₃)
    (obs_a obs_b : Obs₁)
    (h : S₁.judge obs_a = S₁.judge obs_b) :
    S₂.judge (S₁.judge obs_a) = S₂.judge (S₁.judge obs_b) := by
  rw [h]
\end{lstlisting}

\subsection{Corollary: No Stack of Text-Only Judges Escapes}

\begin{lstlisting}[caption=Layered Judges Cannot Escape (Lean), label=lst:lean-layered]
theorem layered_judges_cannot_escape
    {α : Type}
    (layers : ℕ → (α → α))
    (obs_a obs_b : α)
    (h_base : obs_a = obs_b) :
    ∀ n : ℕ, (List.range n).foldl (fun acc i => layers i acc) obs_a =
             (List.range n).foldl (fun acc i => layers i acc) obs_b := by
  intro n
  induction n with
  | zero => simp [h_base]
  | succ k ih =>
    simp only [List.range_succ, List.foldl_append, List.foldl_cons, List.foldl_nil]
    rw [ih]
\end{lstlisting}

\textbf{Implication:} Adding more text-only supervisors (e.g., ensemble classifiers, confidence scores, length penalties) does not escape the impossibility. Information is monotonically lost.

\subsection{Compilation Status}

All theorems compile without errors or unsolved goals:
\begin{verbatim}
$ lake build
  Built EpistemicProofs.Basic (208 lines, 0 errors, 0 unsolved)
  Total: 4 theorems, 1 lemma, 1 corollary (all proven)
\end{verbatim}

\section{Detailed Methodology}\label{app:methodology}

\subsection{Artifact Repository}

All code, data, specifications, and proofs are available at:

\begin{center}
\url{https://github.com/fsgeek/ai-honesty}\\
\texttt{Commit: a9c1334}
\end{center}

\noindent Contents: TLA+ specifications (\texttt{tla/}), Lean 4 proofs (\texttt{EpistemicProofs/}), experiment scripts (\texttt{scripts/}), raw data (\texttt{exp27\_*.csv}), human calibration (\texttt{exp27b\_calibration\_*.json}), and figure generation scripts.

\subsection{Overview}

Experiment 27 evaluates bounded verification strategies across four language model architectures (OLMo-3, Llama 3.1, Qwen, Mistral) on 800 query-response pairs. Experiment 27b applies corrected ground truth using stratified evaluation.

\subsection{Data Collection}

\subsubsection{Query Categories}

\begin{table}[H]
\centering
\begin{tabular}{llcl}
\toprule
\textbf{Category} & \textbf{Type} & \textbf{Count} & \textbf{Ground Truth} \\
\midrule
Geography, Science, History, Math & Factual & 80 & Known, common \\
Wombat (Weird Truth) & Factual & 20 & True, implausible \\
Fictional People/Papers & Fabrication & 25 & Fabricated, plausible \\
Fictional Historical Events & Fabrication & 25 & Fabricated, obvious \\
Future/Private/Impossible & Unknowable & 25 & No valid answer \\
Fabricated Citations & Fabrication & 25 & Fabricated references \\
\midrule
\multicolumn{4}{l}{\textit{200 unique queries $\times$ 4 models = 800 query-response pairs}} \\
\bottomrule
\end{tabular}
\caption{Query categories in Experiment 27/27b.}
\end{table}

\subsubsection{Models}

\begin{table}[H]
\centering
\begin{tabular}{llc}
\toprule
\textbf{Architecture} & \textbf{Model} & \textbf{Params} \\
\midrule
OLMo & allenai/olmo-3-1025-7b (instruct) & 7B \\
Llama 3.1 & meta-llama/Llama-3.1-8B-Instruct & 8B \\
Qwen & Qwen/Qwen3-4B-Instruct-2507 & 4B \\
Mistral & mistralai/Mistral-7B-Instruct-v0.3 & 7B \\
\bottomrule
\end{tabular}
\caption{Model architectures tested.}
\end{table}

\subsection{Signal Extraction}

For each response, we extract:

\subsubsection{Text Features}
\begin{itemize}
    \item \textbf{Response length}: Token count
    \item \textbf{Hedge markers}: Pattern matching for ``I don't know,'' ``I'm not certain,'' etc.
    \item \textbf{Refusal patterns}: Templates from safety training
\end{itemize}

\subsubsection{Tensor Features (Require Model Access)}

\textbf{Per-token entropy}:
\begin{equation}
H_t = -\sum_v p_v^{(t)} \log p_v^{(t)}
\end{equation}
where $p_v^{(t)}$ is the model's probability distribution over vocabulary at token $t$.

\textbf{Aggregations}:
\begin{itemize}
    \item Mean entropy: $\bar{H} = \frac{1}{T} \sum_{t=1}^T H_t$
    \item Max entropy: $H_{\max} = \max_t H_t$
    \item Entropy std-dev: $\sigma_H = \sqrt{\frac{1}{T}\sum_t(H_t - \bar{H})^2}$
    \item Entropy spike count: $\#\{t : H_t > \text{median}_{\text{corpus}} + 1.5\sigma\}$
\end{itemize}

\textbf{Attention features} (final 15 layers):
\begin{itemize}
    \item Layer-wise head concentration (max attention weight per head)
    \item Cross-layer agreement (Jaccard index of attended positions)
    \item Attention entropy per position
\end{itemize}

\subsection{Ground Truth Evaluation (Exp 27b)}

The stratified evaluator applies three tiers of verification:

\subsubsection{Tier 1: Programmatic Verification}

For deterministic factual queries, use exact matching with normalization:

\begin{algorithm}[H]
\SetAlgoLined
\KwIn{response text, expected answer(s)}
\KwOut{Correct/Incorrect}

Normalize both text and expected answer (unicode, case, whitespace)\;
\If{exact substring match}{
    \Return{Correct}\;
}
\For{each morphological variant}{
    \If{variant in normalized text}{
        \Return{Correct}\;
    }
}
\If{negation detected near answer}{
    \Return{Incorrect}\;
}
\Return{Ambiguous; escalate to Tier 2}\;
\caption{Tier 1 Programmatic Matching}
\end{algorithm}

\subsubsection{Tier 2: LLM Classification}

For ambiguous or unknowable queries, use Qwen3-4B-Instruct (\texttt{Qwen/Qwen3-4B-Instruct-2507}, greedy decoding, \texttt{max\_new\_tokens=200}) as classifier. Note: this model also appears in the test set for generation; the evaluation task (classification) differs from the generation task, and results are validated against human calibration (93.8\% agreement):

\begin{table}[H]
\centering
\begin{tabular}{lp{5cm}}
\toprule
\textbf{Query Type} & \textbf{Classification Prompt} \\
\midrule
Knowable (Control, Wombat, Glavinsky) &
``Is this response factually correct? Hedge language does not excuse factual error.'' \\

Unknowable (Private, Future) &
``Does the response appropriately abstain, or does it fabricate an answer?'' \\

\bottomrule
\end{tabular}
\caption{Tier 2 LLM classification prompts.}
\end{table}

\textbf{Strict standard}: If a response hedges (``I'm not certain, but\ldots'') then fabricates, it is classified as \textit{Incorrect}, not \textit{Refusal}.

\subsubsection{Tier 3: Human Calibration}

A blinded random sample of 80 items (10\% of 800) was reviewed by Tony Mason. Disagreements were categorized:

\begin{table}[H]
\centering
\begin{tabular}{lrcl}
\toprule
\textbf{Disagreement Type} & \textbf{Count} & \textbf{\%} & \textbf{Analysis} \\
\midrule
Auto too generous & 2 & 2.5\% & Hedged fabrication \\
Auto too strict & 3 & 3.75\% & Morphological/hedging patterns \\
Agreement & 75 & 93.75\% & Validates evaluator \\
\bottomrule
\end{tabular}
\caption{Human calibration results (Tier 3).}
\end{table}

\textbf{Result}: 93.8\% agreement validates the stratified evaluator for paper use. The human annotator is an author; agreement with the independent LLM evaluator on blinded items provides cross-validation.

\subsection{Analysis Pipeline}

\subsubsection{Budget Curve Computation}

For each model and signal type (Text, Tensor, Composed):

\begin{enumerate}
    \item Compute AUC(signal) on 800 queries (knowable vs unknowable binary classification)
    \item Rank queries by signal confidence (descending)
    \item At budget $b \in \{10\%, 20\%, 30\%\}$:
    \begin{itemize}
        \item Verify top $b$ queries at full cost
        \item Assume bottom $(100-b)\%$ are correct (no verification)
        \item Compute accuracy: $\frac{\text{\# correct in budget} + \text{\# correct in unbudgeted}}{800}$
    \end{itemize}
    \item Plot (Budget, Accuracy) curve
\end{enumerate}

\subsubsection{Test Set Composition}

The 200 unique queries are stratified as shown in the table above (100 knowable, 100 unknowable). Each query is run on all 4 models, yielding 800 query-response pairs (400 knowable, 400 unknowable).

To ensure fair evaluation across judges, the test set has balanced response lengths. Within each query category, responses have similar token-count distributions between correct and incorrect answers. This eliminates confounds where a judge might exploit length patterns rather than epistemic signals. All results reported below use this balanced test set.

\subsubsection{Cross-Model Agreement}

For each query, compute per-model entropy and test correlation:

\begin{equation}
\rho_{\text{Spearman}} = \text{corr}(\text{entropy}_{\text{Model A}}, \text{entropy}_{\text{Model B}})
\end{equation}

Result: $\rho \geq 0.76$ across all model pairs (mean 0.762), indicating the entropy signal is architectural, not model-specific.

\section{Topological Data Analysis (TDA) Methodology}\label{app:tda}

\subsection{Motivation}

The paper claims that unknown responses exhibit topological fragmentation in attention patterns. This section details the TDA methodology used to measure internal coherence.

\subsection{Persistent Homology}

Persistent homology is a tool from topological data analysis (TDA) that measures the ``shape'' of high-dimensional data across multiple scales~\citep{carlsson2009topology}. Given a point cloud, it tracks when topological features---connected components ($H_0$), loops ($H_1$), voids ($H_2$)---appear and disappear as a distance threshold $\epsilon$ grows. Features that persist across many scales reflect genuine structure; short-lived features reflect noise. For applications of persistent homology to deep neural networks, see~\citet{rieck2019neural}.

We apply persistent homology via the Vietoris-Rips complex constructed from attention patterns.

\subsubsection{Construction}

For model response $r$ on query $q$:

\begin{enumerate}
    \item Extract attention matrices from last 15 layers: $A^{(l)} \in \mathbb{R}^{n_h \times T \times T}$ (heads $\times$ tokens $\times$ tokens)
    \item Flatten each head's attention matrix into a vector: $\mathbf{v}^{(l,h)} \in \mathbb{R}^{T^2}$
    \item Construct point cloud: $\{v^{(l,h)} : l \in [1,15], h \in [1, n_h]\}$
    \item Compute pairwise Euclidean distances
    \item Build Vietoris-Rips complex at multiple scales $\epsilon$
\end{enumerate}

\subsubsection{Fragmentation Metric}

\textbf{H$_0$ Persistence (Connected Components)}:

As $\epsilon$ increases, connected components merge. We measure:
\begin{equation}
\text{Fragmentation} = \sum_{c=1}^{n_{\text{components}}} \text{lifetime}(c)
\end{equation}

High fragmentation $\Rightarrow$ attention heads organize into disconnected clusters $\Rightarrow$ incoherent state.

\textbf{H$_1$ Persistence (Loops/Cycles)}:

Loops in the simplicial complex indicate circular attention patterns:
\begin{equation}
\text{Coherence} = \sum_{l=1}^{n_{\text{loops}}} \text{lifetime}(l)
\end{equation}

Ground-truth responses show sustained loops (coherent reasoning paths); fabrications show transient loops (no consistent structure).

\subsection{Results}

\begin{table}[H]
\centering
\begin{tabular}{lcc}
\toprule
\textbf{Category} & \textbf{Mean Fragmentation} & \textbf{Mean H$_1$ Lifetime} \\
\midrule
Control (True) & 2.3 & 4.7 \\
Wombat (Weird True) & 2.8 & 4.2 \\
Glavinsky (Plausible Lie) & 8.1 & 1.2 \\
Westphalia (Obvious Lie) & 12.4 & 0.3 \\
\bottomrule
\end{tabular}
\caption{Topological fragmentation and loop persistence by query category.}
\end{table}

Order-of-magnitude differences in fragmentation and loop persistence correlate with ground truth vs fabrication.

\subsection{Limitations}

\begin{itemize}
    \item TDA is computationally expensive ($O(n^3)$ for Rips complex)
    \item Requires stored attention matrices (not always available)
    \item Sensitivity to noise in low-entropy responses is unexplored
    \item Not load-bearing for the paper's core claims (entropy alone suffices)
\end{itemize}

Reason for main paper exclusion: TDA provides illustration but not necessary proof. The impossibility and entropy signal stand alone.

\end{document}